\documentclass[11pt]{article}
\usepackage{verbatim,amssymb,amsthm,amsfonts}
\usepackage{xfrac}
\usepackage[nonamelimits]{amsmath}
\usepackage{fullpage}
\usepackage[dvips,colorlinks=true,linkcolor=black, citecolor=black]{hyperref}

\def\B{\{0,1\}}

\def\Xbar{\overline{X}}

\def\hbar{\overline{h}}

\providecommand\abs[1]{\lvert#1\rvert}
\providecommand\bigabs[1]{\bigl\lvert#1\bigr\rvert}
\providecommand\ip[1]{\langle#1\rangle}

\def\poly{\mathrm{poly}}

\theoremstyle{plain}

\newtheorem{theorem}{Theorem} % [section]
\newtheorem{lemma}[theorem]{Lemma}

\newtheorem{proposition}[theorem]{Proposition}
\newtheorem{claim}[theorem]{Claim}

\theoremstyle{definition}

\def\hypref#1#2{\hyperref[#2]{#1~\ref*{#2}}}
\def\nc{\mathrm{NC}}

\DeclareMathOperator\pr{\mathrm{Pr}}
\DeclareMathOperator\E{\mathrm{E}}

\def\eps{\varepsilon}

\def\calD{\mathcal{D}}

\def\calH{\mathcal{H}}

\title{Sparse extractor families for all the entropy}
\author{Andrej Bogdanov\thanks{{\tt andrejb@cse.cuhk.edu.hk}. Department of Computer Science and Engineering and Institute for Theoretical Computer Science and Communications, Chinese University of Hong Kong. Work supported by grants RGC GRF CUHK410309 and CUHK410111.} \and Siyao Guo\thanks{{\tt syguo@cse.cuhk.edu.hk}. Department of Computer Science and Engineering, Chinese University of Hong Kong.}}
\date{}

\begin{document}
\maketitle

\begin{abstract}
We consider the problem of extracting entropy by sparse transformations, namely functions with a small number of overall input-output dependencies. In contrast to previous works, we seek extractors for essentially all the entropy without any assumption on the underlying distribution beyond a min-entropy requirement. We give two simple constructions of {\em sparse extractor families}, which are collections of sparse functions such that for any distribution $X$ on inputs of sufficiently high min-entropy, the output of most functions from the collection on a random input chosen from $X$ is statistically close to uniform.

For strong extractor families (i.e., functions in the family do not take additional randomness) we give upper and lower bounds on the sparsity that are tight up to a constant factor for a wide range of min-entropies. We then prove that for some min-entropies weak extractor families can achieve better sparsity.

We show how this construction can be used towards more efficient parallel transformation of (non-uniform) one-way functions into pseudorandom generators. More generally, sparse extractor families can be used instead of pairwise independence in various randomized or nonuniform settings where preserving locality (i.e., parallelism) is of interest.
\end{abstract}

% For the upper bound, we show that a random matrix over $GF(2)$ whose entries are independent but appropriately biased is a strong extractor. For the lower bound, we show that extracting randomness from a (truncated) $p$-biased distribution (each bit is $1$ with probability $p$) over $\B^n$, where $p$ is chosen at random from some small range, requires a somewhat large number of input-output dependencies.

\thispagestyle{empty}
\newpage
\setcounter{page}{1}

\section{Introduction}
Randomness extractors~\cite{NZ93, Sha04}, have numerous applications in complexity theory and cryptography. For example, in computational complexity they are used to isolate satisfying assignments of boolean formulas~\cite{VV86} and in constructing pseudorandom generators for space-bounded computation~\cite{Nis91, RR99}. One notable use in cryptography is in the construction of pseudorandom generators from one-way functions~\cite{HILL, HRV, VZ}.

In many applications of extractors, including the above ones, it is important that the extractors recover essentially all the entropy of the input distribution. A popular choice in such scenarios is to instantiate the extractor by a pairwise-independent hash function family~\cite{BBR88, ILL89}. Pairwise-independent functions are appealing because they have a variety of implementations, ranging from very simple ones~\cite{CW} to very efficient ones~\cite{IKOS}.

Mansour, Nisan, and Tiwari~\cite{MST} observed that pairwise-independent hash functions must be ``dense'' in the sense that a typical output in a typical function in the family must depend on a linear number of inputs. So despite their numerous nice properties, in terms of the number of input-output dependencies, pairwise-independent functions are quite complex. Motivated by an application to local cryptography, Bogdanov and Rosen~\cite{BR11} recently gave a way to bypass this barrier in the context of hardness amplification of ``local'' functions.

In this work we study sparse extractors for all the entropy, these are extractors with a small number of overall input-output dependencies. We consider the more general notion of sparse extractor families. An {\em extractor family} for distributions of min-entropy $k$ over $\B^n$ with error $\eps$ is a distribution $H$ on functions $\B^n \times \B^s \to \B^m$ where $m \leq s + k$ such that for every distribution $X$ over $\B^n$ of min-entropy $k$, the statistical distance between $(H, H(X, U_s))$ and $(H, U_m)$ is at most $\eps$ (where $U_s$ and $U_m$ are uniformly random). The extractor family is {\em strong} if $s = 0$, i.e. $H$ does not take any additional randomness beyond $X$.

Without the sparsity restriction extractors and extractor families are essentially the same object, as the randomness used to choose an extractor from the family can be included in the seed. Once we take sparsity into consideration, however, extractor families allow for more flexibility.  This advantage is especially pronounced in the case of strong extractors: Any single strong extractor in which some output bit depends only on $\ell$ input bits cannot extract from a source that fixes all those $\ell$ bits. In contrast, we show strong extractor families can achieve much better sparsity. 

In this work we prove three results regarding sparse extractor families. First we give a simple construction of sparse extractor families for all the entropy. Then we show that the sparsity of our construction is optimal up to constant factors for a wide range of the min-entropy parameter. Finally, we show that an equally simple construction of weak extractor families achieves better sparsity. Thus when sparsity is required, weak extractors can provably outperform strong ones. 

We also show our weak extractor family gives a somewhat improved nonuniform construction of local pseudorandom generators from local one-way functions, based on recent work of Vadhan and Zheng~\cite{VZ}.  In general our results can be useful in randomized or nonuniform settings where hashing is used and obtaining or preserving small input-output dependencies (i.e., parallelism) is of interest.

\subsection{Our results}

Let $h\colon \B^n \to \B^m$ be a function. We say output $j$ of $h$ {\em depends} on input $i$ if there exists assignments $x, x' \in \B^n$ that differ only in the $i$th coordinate such that $h(x)_j \neq h(x')_j$. We say $h$ is {\em $s$-sparse} if the number of input-output pairs $(i, j)$ such that output $j$ depends on input $i$ is at most $s$, and $h$ is {\em $\ell$-local} if every output $j$ depends on at most $\ell$ inputs $i$. 

\begin{theorem}
\label{thm:upper}
Let $K$ be a sufficiently large constant and $n, k, m, \delta$ be parameters such that $1 \leq m\leq k \leq n$, and $0 < \delta < 1$.  Let $H(x) = Mx$, where $M$ is an $m\times n$ matrix over $GF(2)$ where each entry equals $1$ independently with probability
       \[ p = \min\Bigl\{\frac{1}{m} \cdot \log\frac{m}{\delta}\ln\frac{Kn}{m} , \frac12 \Bigr\}. \]
Then $H$ is a strong extractor family for min-entropy $k$ with error at most $\sfrac12\sqrt{\delta + K\cdot\/2^{- k+m}}$.
% In the case of $2m\log{(m/\delta)}\leq k$, we can improve $p$ to be
%	$ p = \min\Bigl\{\frac{1}{k} \cdot \log(m/\delta)\ln(Kn/k) , \frac12 \Bigr\} $
% so that $\calH$ is a strong extractor for min-entropy $k$ with error at most $\sfrac12\sqrt{\delta }$.
\end{theorem}

By a large deviation bound, all but an $\delta$-fraction of $H$ are $O(nmp)$-sparse.  The best error we can hope for~\cite{RT} is  $\Omega(\sqrt{2^{-k+m}})$, which is achieved by a pairwise independent hash function family. Perhaps the simplest construction of such a family is to choose each entry of $M$ independently at random with probability $p = 1/2$. When $p<1/2$, Theorem~\ref{thm:upper} shows the sparsity can be reduced dramatically at the cost of increasing the error by a little. For example if we set $\delta = 2^{-k+m}$, we obtain a $O(n \log{m} \log{(n/m)})$ sparse strong extractor family whose error is within a constant factor of optimal.

% The second item $K 2^{-k+m}$ is tight to the best case by a constant factor so that it is somehow unavoidable. The first item $\delta$ is directly related with locality and  usually we want  $\delta=O(2^{-k+m})$ so that the whole error is  close to the best case.  However, as $k-m$ becoming larger,  $\delta$ can't always be $O(2^{-k+m})$  due the constraint $p<1/2$. Hence $\delta$ will dominate the error when entropy loss $k-m$ is large.  To extract almost all the entropy which is the case $k-m=O(1)$, we can just let $\delta=2^{-k+m}$.  So that we obtain an $O((n/m) \log{m} \log{(n/m)})$ local strong extractor family with constant error. 
%In particular, in the case  the source has constant rate entropy, say $m=k-O(1)=\Omega(n)$, the locality is $O(\log{n})$.  

Our main negative result shows that this sparsity is necessary for a large range of values of $k$ and when $\delta$ is constant. 

\begin{theorem}
\label{thm:lower}
Suppose $n^{0.99} \leq m \leq n/6$. There exists a distribution $\calD$ over distributions $\Xbar$ on $\B^n$ of min-entropy $1.5m$ each so that for every function $h\colon \B^n \to \B^m$ of sparsity $0.001 n \log m \log(2n/m)$, the expected statistical distance between $h(\Xbar)$ and the uniform distribution over $\B^m$ is at least $1 - e^{-m^{\Omega(1)}}$.
\end{theorem}

Applying Yao's minimax principle or just take the covex combination of distributions $\overline{X}$ as the bad distribution, we conclude that  the sparsity in Theorem~\ref{thm:upper} is optimal up to constant factor for this range of parameters.  The next results concerns weak extractor families.

\begin{theorem}
\label{thm:upper2}
Let $K$ be a sufficiently large constant and $n, k, c, m$ be parameters such that $1\leq k \leq n$,  $1\leq s<m$ and $c > 1$.  Let $H\colon \B^n \times \B^s \to \B^m$ be given by $H(x, r) = Mx + Br$, where $M$ is an $m\times n$ random matrix in which each entry equals $1$ independently with probability
        \[ p= \min \Bigl\{\frac{K}{m} \cdot \ln\frac{n}{\ln c},\frac12 \Bigr\}, \]
and $B$ is an $m\times s\ (m>s)$ matrix  of full rank where every set of at most $m/2K$ rows is linearly independent. Then $H$ is an extractor family for min-entropy $k$  with error $\sfrac12\sqrt{c \cdot 2^{-k-s+m}}$.
\end{theorem}

The construction of a matrix $B$ with the desired properties and $O(m)$ sparsity is a well studied problem in the theory of low density parity check codes~\cite{Ga62,SS96}. Capalbo et al.~\cite{CRVW} give an explicit construction with $s = \alpha m$ for some constant $\alpha < 1$ and every $m$, which is optimal up to the choice of the constant $\alpha$. Instantiating Theorem~\ref{thm:upper2} with this matrix, and setting $c = 2$, we obtain a family of $O(n \log n)$  sparse extractors with error $O(\sqrt{2^{-k-s+m}})$, which is optimal up to constant factor. (If $m = k + s - O(1)$, the output contains almost all the entropy from the source plus all the entropy invested by the seed and the error is an arbitrarily small constant.) By using a larger value of $c$, we can reduce the sparsity at the cost of increasing the error.

We observe that using the randomized encoding of Applebaum et al.~\cite{AIK}, these extractors can be made to have constant locality at the cost of increasing the seed length $r$ to $O(n \log n)$ bits.

For certain parameters the weak extractor family from Theorem~\ref{thm:upper2} bypasses the limitation on strong extractor families from Theorem~\ref{thm:lower}. For example, for constant statistical distance and $m = n^{0.99}$ Theorem~\ref{thm:lower} implies that for a strong extractor to produce even a constant fraction of the entropy a sparsity of $\Omega(n (\log n)^2)$ is necessary, while Theorem~\ref{thm:upper2} says that an $O(n \log n)$-sparse weak extractor family can extract all but $O(1)$ bits of entropy.  

\subsection{An application}

\paragraph{Pseudorandom generators from one-way functions}  The construction of pseudorandom generators from one-way functions of H\aa\/stad et al.~\cite{HILL} does not in general preserve locality.  Haitner, Reingold, and Vadhan~\cite{HRV} gave a construction that is more efficient and can be implemented in $\nc^1$. Recently Vadhan and Zheng~\cite{VZ} gave an even simpler variant of this construction. In combination with the ``compiler'' of Applebaum, Ishai, and Kushilevitz ~\cite{AIK04}, one obtains a generic locality-preserving transformation of one-way functions into pseudorandom generators. 

Applying the transformation of Applebaum et al.  may have an adverse effect on seed length, as it may grow quadratically. However, the construction of Vadhan and Zheng is extremely simple; it is obtained by applying extractor to a sequence of ``blocks'', each of which inherits the locality of the one-way function $f$. Instantiating the extractor by the construction from Theorem~\ref{thm:upper2}, we obtain a transformation of {\em nonuniform} one-way functions into nonuniform pseudorandom generators that preserves output locality logarithmic in the size of the adversary with the same seed length as the one obtained by Vadhan and Zheng. Using an additional idea of Applebaum et al., the transformation can be made to preserve constant output locality at the expense of increasing the seed length. We describe this application in Section~\ref{sec:prg} (see Proposition~\ref{prop:owfprg}).

\subsection{Related work}

%%  Siyao:  Vadhan's construction can help us to construct sparse strong extractor with optimal locality (nm/k\logm\logn/k).   He firstly randomly sample m+\eps bits from inputs, showing the m+\eps bits containing sufficient entropy by this way, after that apply an extractor on the m+\eps bits. We can apply our sparse extractor (for large m) on this m+\eps bits thus get a sparse extractor with desired sparsity.

\paragraph{Sparse extractors for restricted sources} Motivated by certain applications, Zhou and Bruck ~\cite{ZB11} show that low density random matrices can efficiently extract random bits from some restricted noisy sources, such as bit fixing sources and Markov sources. Our Theorem~\ref{thm:upper} shows that essentially the same construction extracts from arbitrary sources of given min-entropy.

\paragraph{Extractors in $\nc^0$}  Applebaum, Ishai and Kushilevitz~\cite{AIK06} give a weak extractor in $\nc^0$ (thus sparsity $O(n)$) works for min-entropy $k = (1 - O(1))n$, but suffers $\Omega(n)$ entropy loss. Our extractor family from Theorem~\ref{thm:upper2} matches these parameters. The construction from~\cite{AIK06} does not appear to extend to distributions of smaller min-entropy or allow for smaller entropy loss, while ours does. However, they provide a single extractor that works for all distributions, while we only give an extractor family.

%\paragraph{Extracting by shallow circuits} Viola~\cite{Vio05} showed the existence of a distribution over distributions $X$ over $\B^n$ of min-entropy $k$ each so that no circuit of depth $d$ and size $2^{O(n/k)^{1/(d-1)}}$ can extract a constant fraction of the entropy of $X$ (unless it uses seed of length linear in the output). 

%($\bold{Siyao}$: can we say this for sparse function?) Since $\ell$-local functions are computable by circuits of depth two and size $2^\ell$, it follows that strong extractor families require sparsity $\ell = \Omega(n)$. Our Theorem~\ref{thm:lower} strengthens this lower bound by the factor of $\Omega(\log k \log(n/k))$.

 %Our hard distributions $X$ have the same form as the ones in~\cite{Vio05}; however to obtain a stronger bound we rely on a more refined probabilistic analysis, which takes advantage of small sparsity and not merely small depth.
 
% Viola ~\cite{Vio05} showed extractors are not computable by small depth circuits when seed length is smaller than $n$. We can extend his argument to show the seed length of sparse extractors should be larger than $n$. His argument is not sufficient to derive  Theorem~\ref{thm:lower}  but his argument helps us to simplify our original proof for non-linear case.  

% difference between our problem and locally computable extractors
%%($\bold{Siyao}$: Should we delete the paragraph of locally computatble extractor?)

\paragraph{Locally computable extractors} A {\em locally computable extractor}~\cite{Lu04, Vad03} is an extractor in which after the seed is fixed, the output as a whole depends on a small number of input bits. Such extractors are used to implement private-key encryption in the bounded storage model~\cite{Maurer}. We observe that the notions of locally computable extractors and sparse extractor families are fundamentally different. This is best illustrated in the regime in which we extract all the entropy, which is of main interest in this work. A lower bound of Vadhan~\cite{Vad03} shows that when the output length $m$ is linear in the min-entropy $k$, then even after the seed is fixed the output of the extractor as a whole must depend on at least a linear fraction of the input. Thus $o(n)$-locally computable extractors are not possible when $m = \Omega(k)$. Although this is inevitable, our results show that it is possible to make a small number of input-output dependencies.

We observe that the locally computable extractors of Lu~\cite{Lu04} and De and Trevisan~\cite{DT09} are also sparse, but they extract only a fixed root of the min-entropy $k$. This is sufficient for bounded storage cryptography, but not for the application we describe in Section~\ref{sec:applications}.

\iffalse
\begin{table}[position specifier]
\centering
\begin{tabular}{|l|l|l|l|l|l|}
\hline
       & locality  & entropy  & entropy loss &  error &  remark\\
\hline
 ~\cite{AIK06}   &  $O(1)$  & $k=(1-O(1))n$ & $\Omega(n)$ &  $2^{-\Omega(n)}$ & $m=n$ \\
Theorem ~\ref{thm:upper}    &  $O((n/m)\log{(n/m)}\log{(m/\delta)})$  &  $1\leq k\leq n$ &  $k-m$ & $1/2\sqrt{\delta+O(2^{-k+m})} $ &\\
Theorem ~\ref{thm:upper2}    &  $O((n/m)\log{(n/\ln(1/\delta))})$  &  $1\leq k\leq n$ &  $k+s-m$ & $1/2\sqrt{(1/\delta)\cdot 2^{-k+m}} $ &  $m>s$\\
\hline
\end{tabular}
\caption{Upper bound}
\label{tab: upper bound}
\end{table}

\begin{table}[position specifier]
\centering
\begin{tabular}{|l|l|l|l}
\hline
    & locality & Remark\\
\hline
~\cite{Vio05}   &  $\Omega(n/k)$  & \\
Theorem ~\ref{thm:lower}   & $\Omega((n/k)\log(n/k)\log{k})$   & $n^{0.99}\leq m\leq n/6, $\\
\hline
\end{tabular}
\caption{Lower bound: $m=\Omega(k), \eps=\Omega(1)$}
\label{tab:lower bound}
\end{table}
\fi

\subsection{Our proofs}
Bogdanov and Rosen~\cite{BR11} proved a quantitatively weaker version of Theorem~\ref{thm:upper} that achieves sparsity $O(n (\log n)^3)$ instead of the optimal $O(n\log(m)\log(n/m))$. (They did not attempt to determine the dependence on $m$ and they can achieve sparsity $O(n\log(m)\log(n/m)^2)$) In their proof, $H$ is viewed as a collection of boolean functions $(h_1, \dots, h_m), h_i\colon \B^n \to \B$. They show that for most choices of $h_1$, conditioning on $h_1(x)$ reduces the min-entropy $k$ of $x$ by at most $1 + 1/\poly(k)$ bits (unless $k$ is very small), and so this bit extraction can be applied iteratively for $m$ steps.

One drawback of this argument is that as $i$ gets larger and the min-entropy of $x$ conditioned on $h_1(x),\dots,h_{i-1}(x)$ becomes smaller, the density of the functions $h_i$ must keep increasing (as required by our lower bound). To achieve our optimal (up to constant) sparsity, we must analyze the effect of all the functions $h_1, \dots, h_m$ simultaneously. 

To do this, we upper bound the probability that two samples $x, x'$ collide under $h$, that is the probability that $h(x + x') = 0$. For a fixed pair $(x, x')$, each entry $h_i(x + x')$ of $h(x + x')$ is biased towards zero. We can think of $h_i(x + x')$ as a random variable that takes value zero with some probability $p(x, x')$, and is unbiased otherwise. Intuitively, our analysis shows that the unbiased components of this distribution dominate in collisions. Several technical complications arise in the formal argument. One useful tool that allows us to analyze the case when most of the components of $h(x + x')$ are unbiased is H\"older's inequality. % We give the proof of Theorem~\ref{thm:upper} in Section~\ref{sec:upper}.

To give an idea of our proof of Theorem~\ref{thm:lower}, let's make the simplifying assumption that $h$ is $\ell$-local, where $\ell = \gamma (n/m) \log m \log(n/m)$. We give a heuristic argument why we expect the output of $h$ to be far from uniform when $h$ is linear. Let $X$ be the $p$-biased distribution over $\B^n$ (each bit takes value $1$ independently with probability $p$) and $p$ is chosen so that $H(p) = m/n$, where $H(p)$ is the binary entropy of $p$. Then the distribution $X$ has Shannon entropy $m$. However, every output bit of $h(X)$ is $(1 - 2p)^\ell$-biased, and we chose the parameters so that $(1 - 2p)^\ell = m^{-\Omega(\gamma)}$. By choosing $\gamma$ small enough, we can ensure that every output bit of $h$ has, say, $m^{-1/2}$ bits of {\em entropy deficiency}, so by the sub-additivity of Shannon entropy $h(X)$ has $m^{1/2}$ fewer bits of entropy than a uniformly random variable over $\B^m$. So $h(X)$ does not ``look'' random in terms of Shannon entropy.

To turn this heuristic argument into a proof we need to handle several issues, the most interesting of which is replacing entropy deficiency by statistical distance from the uniform distribution. One advantage of measuring entropy deficiency is that entropy is subadditive, which allows us to ignore the dependencies between the various outputs of $h(X)$ in the above argument. In contrast, to obtain a good lower bound on statistical distance we must take into account these dependencies. Here we apply tail bound for read $t$ family ~\cite{gavinsky2012tail}.  To extend the analysis from linear functions to general ones we apply an elegant idea of Viola~\cite{Vio05} of shifting $X$ by a random offset.

We establish Theorem~\ref{thm:upper2} by a relatively straightforward probabilistic calculation.

\subsection{Open problems}

In terms of seed length our sparse extractors are quite poor. For example the size of the family $H$ in Theorem~\ref{thm:upper} is exponential in $n$. By a standard probabilistic argument it can be shown that a random sample of $H$ of size $O(n / \eps^2)$ is as effective as the whole family while incurring an additional penalty of only $\eps$ in statistical distance. Consequently there is no existential obstacle to sparse extractor families with short seed. It remains to see if such families can be found efficiently.

For weak sparse extractors, our work leaves open two possible improvements. First, we do not know if the sparsity of our weak extractors is the best possible. Second, we do not know what is the minimal size of a sparse weak extractor family. It could be that even a family of size 1, i.e. a single sparse weak extractor, is sufficient. Such an extractor could be used to obtain a uniform construction of local pseudorandom generators from local one-way functions. Could there be a single sparse weak extractor of sparsity linear in $n$ that extracts $k - k^{0.99}$ bits of min-entropy for every source over $\B^n$ of min-entropy $k$?

\section{Proof of Theorem~\ref{thm:upper}}
\label{sec:upper}

To prove Theorem~\ref{thm:upper} it is sufficient to show that for every set $S$ of size $2^{k}$, the statistical distance between $(H, H(X))$ and $(H, U)$ is at most $\sfrac12\sqrt{\delta + O(2^{-k+m})}$, where $X$ is chosen at random from $S$ and $U$ is uniformly random.
In fact we will show for every $x_0 \in S$,
\[ \pr_{H, X}[H(X)=H(x_0)] \leq \frac{1+ \delta + O(2^{-k+m})}{2^m} \]
from where
\[ \pr_{H, X, X'}[H(X) = H(X')] \leq \max_{x_0 \in S} \pr_{H, X}[H(X) = H(x_0)] \leq \frac{1+ \delta + O(2^{-k+m})}{2^m} \]
where $X$ and $X'$ are independent samples from $S$. This is sufficient to establish Theorem~\ref{thm:upper} using the relation between collision probability and statistical distance from Claim~\ref{claim:colldist} in Appendix~\ref{app:colldist}.

\begin{proof}
When $p = 1/2$ the analysis is standard, so we will assume that $p =\frac{1}{m} \cdot \log(m/\delta)\ln(15 n/m) < 1/2$.
Since entries of $M$ are chosen independently from each other, for any $y \in \{0,1\}^n$, we have
\[
\pr_H[H(y) = 0] = \pr_a[\ip{a, y} = 0]^m = \Bigl(\frac{1+(1-2p)^{|y|}}{2}\Bigr)^m =\frac{1}{2^m}\sum_{i=0}^m\binom{m}{i}(1-2p)^{i|y|}
\]
Here $a \sim \B^m$ is chosen from the $p$-biased distribution. Let $S_0$ be the set $\{x_0 + x\colon x\in S\}$. Then
\begin{align*}
\pr_{H, X}[H(X)=H(x_0)] &=\pr_{H, y \sim S_0}[H(y) = 0] \\
  & = \E_{y\sim S_0}\Bigl[\frac{1}{2^m}\sum_{i=0}^m\binom{m}{i}(1-2p)^{i|y|}\Bigr] = \frac{1}{2^m}\sum_{i=0}^m\binom{m}{i}\E_{y\sim S_0}[(1-2p)^{i|y|}] 
\end{align*}
Let $a_i= \E_{y\sim S_0}[(1-2p)^{i|y|}]$. We now upper bound the sum $\sum_{i=0}^m\binom{m}{i}a_i$ by $1+\delta+O(2^{-k+m})$. We will consider two cases: When $i$ is small -- specifically, $i \leq k/(2\log{(m/\delta)})$, we show that $a_i$ decreases at a rate faster than $(m/\delta)^{-i}$, so the sum is dominated by the term $i = 0$. When $i$ is large, we want to bound both $a_i$ and $\binom{m}{i}$ by its largest possible value. To achieve this, we have to further split the large $i$'s into ``mildly large'' and ``very large'' ones and apply the argument to each summation separately. The resulting contribution is $O(2^{-k+m})$.     Notice that, in the case $m\leq k/(2\log{(m/\delta)})$, we need not  consider the contribution of large $i$'s thus we can improve $p$ to be $\frac{1}{k} \cdot \log(m/\delta)\ln(15 n/k)$ by using the same analysis for small $i$'s.

\paragraph{The small $i$'s.} We show that if $i \leq k/(2\log{(m/\delta)})$, then $a_i \leq (m/\delta)^{-1.8i}$ and therefore
\[ \sum_{i=0}^{k/(2\log{(m/\delta)})} \binom{m}{i}a_i\leq (1+(m/\delta)^{-1.8})^m \leq e^{{\delta}^{1.8}/m^{0.8}}\leq 1+\delta. \]
To bound $a_i$ we apply H\"older's inequality, which says that for every $B \geq 1$:
\[ a_i = \E_{y\sim S_0}[(1-2p)^{i|y|}] = \frac{1}{\abs{S}} \sum_{y \in \B^n} 1_{y \in S_0} \cdot (1 - 2p)^{i\abs{y}}
\leq \frac{1}{\abs{S}} \abs{S}^{1-1/B} \sum_{y \in \B^n} \bigl((1 - 2p)^{Bi\abs{y}}\bigr)^{1/B} \]
The last expression can be simplified to give
\[
a_i \leq {\Bigl(\frac{(1+(1-2p)^{Bi})^n}{|S|}\Bigr)}^{1/B} \leq  \Bigl(\frac{(1+e^{-2piB})^n}{|S|}\Bigr)^{1/B}.
\]
We choose $B = k/(2i\log{(m/\delta)})$, which is at least one because $i \leq k/(2 \log(m/\delta))$. By our choice of $p$, it follows that $2piB \geq \ln(15n/m)$ and so
\[ a_i \leq \Bigl(\frac{(1+ m/15n)^n}{|S|}\Bigr)^{1/B} \leq \Bigl(\frac{e^{m/15}}{2^{k}}\Bigr)^{1/B} 
   \leq \Bigl(\frac{e^{k/15}}{2^{k}}\Bigr)^{1/B} \leq (2^{-0.9})^{k/B} = (m/\delta)^{-1.8i}. \]

\paragraph{The large $i$'s.} We have that
\begin{equation}
\label{eqn:basicbound}
a_i = \E_{y\sim S_0}[(1-2p)^{i|y|}] \leq \frac{1}{\abs{S}} \sum_{y \in \B^n} (1-2p)^{i|y|} = \frac{(1 + (1 - 2p)^i)^n}{\abs{S}}
\leq \frac{(1 + e^{-2pi})^n}{\abs{S}}.
\end{equation}
When $i \geq m/4$, the last expression is at most $(1 + e^{-pm/2})^n/\abs{S}$. By our choice of $p$, $pm/2 \geq \sfrac12 \log m \ln(15n/m)$. Optimizing for $\log m$, it can be calculated that this expression is at least $\ln n$ when $n$ is sufficiently large. It then follows that
\[ a_i \leq \frac{(1 + 1/n)^n}{\abs{S}} \leq e 2^{-k} \qquad \text{and so} \qquad \sum_{i=m/4}^m \binom{m}{i} a_i \leq 2^m \cdot e2^{-k} = e2^{-k+m}. \]
Finally, we handle the $i$'s in the range $k/(2\log{(m/\delta)}) < i < m/4$. Using (\ref{eqn:basicbound}) and the lower bound on $i$, we have that
\[ a_i \leq \frac{(1 + m/15n)^n}{\abs{S}} \leq \frac{e^{m/15}}{\abs{S}} \leq 2^{0.1m - k} \]
and so
\[ \sum_{k/(2\log{(m/\delta)}) < i < m/4} \binom{m}{i}a_i \leq 2^{0.1m - k} \sum_{i=0}^{m/4} \binom{m}{i} \leq 2^{0.1m - k} \cdot 2^{H(1/4)m + O(1)} \leq 2^{ - k+m+ O(1)} , \]
where $H$ is the binary entropy function, and $H(1/4) \leq 0.9$.
\end{proof}

\section{Proof of Theorem~\ref{thm:lower}}
\label{sec:lower}
Let distribution $\Xbar$ be a truncated variant of the $p$-biased distribution $X$ where $p$ is chosen so that $H(p) = 2m/n$ and $p \leq 1/2$. The distribution $\calD$ on distributions is defined as follows. Choose $y$ uniformly from $\B^n$ and output $y+\Xbar$.  To prove Theorem~\ref{thm:lower},  we will show for most choices of $y$,  there exists a  statistical test $T_y$ that distinguishes $h(y+\Xbar)$ from the uniform distribution $U$ and then argue the expected statistical distance over choice of $y$ between $h(y+\overline{X})$ and $U$ is large. We will define $T_y$ shortly for $y\in\B^n$ and first argue that for most choice of $y$, $T_y$ distinguishes $h(y+X)$ from $U$. Then we will show how to define $\Xbar$ of min-entropy $1.5m$ in a way that $X$ and $\Xbar$ are statistically close. Finally, we conclude that the expected statistical distance over choice of $y$ between $h(y+\overline{X})$ and $U$ is at least $1-e^{-m^{\Omega(1)}}$.

The following bounds on $p$ are obtained by plugging in $H(p) = 2m/n$ in Lemma~\ref{lemma:entropy} in Appendix~\ref{app:entropy}:
\begin{equation}
\label{eqn:pbounds}
\frac{m}{3n \log_2(n/m)} \leq p \leq \frac{2m}{n \log_2(n/2m)}.
\end{equation}

Now suppose $h$ has sparsity $(m/2p) \beta \log m$, where $\beta$ is a sufficiently small constant, say $\beta = 0.08$. Notice this $\beta$ also satisfies $n \leq m^{1 + 2\beta}$ (since by assumption $m \geq n^{0.99}$).  Partition the inputs of $h'$ into two sets $H$ and $L$, where $H$ contains those inputs that participate in at least $m^{2 - 6\beta}/pn$ outputs of $h'$, and $L$ contains the rest. By Markov's inequality (using the assumption $n \leq m^{1 + 2\beta}$), $H$ has size at most $m^{8\beta} \beta \log k$. For $x \in \B^n$, let $x_0$ and $x_1$ denote its projections onto $H$ and $L$, respectively. For every $y\in\B^n$, we define the statistical test
		\[T_y=\{z\in\B^{m}: \text{$\Delta(h(x_0,y_1),z) \leq 1/2 - m^{-\beta}/4$ for some $x_0$}, \}\]
where $\Delta(a, b)$ is relative Hamming distance between the strings $a$ and $b$, i.e. the fraction of positions in which they differ.

\begin{claim}
\label{claim:ellone}
For sufficiently large $k$, $\Pr_X[h(X + y)\in T_y]\geq 1-e^{-m^{3\beta}/2}$ for at a least $1 - 1-e^{-m^{3\beta}/2}$ values of $y \in \B^n$.
\end{claim}

In the proof we will need the following fact about Boolean functions $f\colon \B^d \to \B$
\begin{equation}
\label{eqn:noise}
\Pr_{X, Y}[f(X+Y)\neq f(Y)] \leq \tfrac12 - \tfrac12(1-2p)^d
\end{equation}
where $Y$ is uniformly distributed in $\B^n$, and $X$ is chosen independently from the $p$-biased distribution on $\B^n$. This fact follows easily by Fourier analysis~\cite{OD02} and was also used by Viola~\cite{Vio05} in a context related to ours. 

We will also make use of the following inequality of Gavinsky et al.~\cite{gavinsky2012tail}. A collection of indicator random variables $Z_1, \dots, Z_m$ is called a {\em read $t$ family} of functions if there exist independent random variables $X_1, \dots, X_n$ such that each $X_i$

Then we will apply tail bound for {\em read t family} of functions ~\cite{gavinsky2012tail} to show for most choice of $(x,y)$ outcome  concentrate on expectation.   Indicator random variables $Z_1, \dots, Z_m$ is a {\em read $t$ family}  if they can written as a function of independent random variables $X_1, \dots, X_n$ where each $X_i$
affects at most $t$ of the $Z_i$'s. Then for every $\eps > 0$,
\begin{equation}
\label{eqn:tail}
\Pr[Z \geq \E[Z] + \eps m] \leq e^{-2\eps^2m/t}.
\end{equation}
where $Z = Z_1 + \dots + Z_m$.

\begin{proof}
We will show that for every choice of $x_0, y_0$, with probability $1 - e^{-\Omega(m^{1-\beta})}$ over the choice of $x_1, y_1$, $h(x+y) = h(x_0 + y_0, x_1 + y_1)$ is in $T_y$.  Fix $x_0, y_0$ and consider the function $h^{x_0+y_0}(x_1) = h(x_0+y_0, x_1)$. Let $Z = Z_1 + \dots + Z_{m}$, where
\[ Z_i = \begin{cases}
1, &\text{if $h^{x_0+y_0}_i(x_1 + y_1)\neq h^{x_0+y_0}_i(y_1)$}, \\
0, &\text{otherwise}.
\end{cases} \]
Suppose $h^{x_0+y_0}_i$ depends on $d_i$ inputs for $1\leq i\leq m$, by (\ref{eqn:noise}) we have
              \[\E[Z_i] \leq 1/2(1-(1-2p)^{d_i}).\] 
By linearity of expectation and $\sum_{i=1}^m d_i\leq (m/2p)\beta\log{m}$,  we get
             \[\E[Z] \leq m/2(1-(1 -2p)^{\beta\log{m}/2p})\leq m(1/2-m^{-\beta}/2).\] 
Now we apply tail bound (\ref{eqn:tail}) to $Z_1, \dots, Z_m$ with $t = m^{2 - 6\beta}/pn$ and $\eps=m^{-\beta}/4$  to obtain
		\[\Pr[Z\geq m(1/2-m^{-\beta}/4)]\leq e^{-2\eps^2m/(m^{2 - 6\beta}/pn)}\leq e^{- m^{3\beta}}\]
where we used the estimate (\ref{eqn:pbounds}) to lower bound $pn$. In other words,
\[\pr_{x_1, y_1}[\Delta(h_{x_0+y_0}(y_1),h_{x_0+y_0} (x_1 + y_1))  \leq 1/2 - m^{ - \beta}/4] \geq 1 -e^{- m^{3\beta}}.\]
 It follows that
\begin{align*}
\pr_{x, y}[h(x+y)\in T_{y}]  
   & \geq \E_{x_0,y_0}\bigl[\pr_{x_1, y_1}[ \Delta(h_{x_0+y_0}( y_1),h_{x_0+y_0}( x_1+y_1))\leq 1/2 - m^{ - \beta}/4]\bigr]   \\ 
 &   \geq 1-e^{-m^{3\beta}}. 
\end{align*}
Applying Markov's inequality, we conclude that for at least $1-e^{-m^{3\beta}/2}$ choices of $y$, 
\[ \E_X[h(X+y)\in T_y]\geq 1-e^{-m^{3\beta}/2}. \hfill\qedhere \]
\end{proof}

\begin{claim}
\label{claim:uniform}
For any fixed $y\in\B^n$, with probability $1-2^{-\Omega(m^{1-2\beta})}$ over the choice of a uniform $U \sim \B^{m}$, $U$ is not in $T_y$.
\end{claim}
\begin{proof}
Since $H$ has size at most $m^{8\beta}\beta \log k$, the range of $h'(x_0, y_1)$ has at most $2^{m^{8\beta}\beta \log m}$ elements. For every  such element $h(x_0, y_1)$, the probability that $U$ is within distance $m/2 - m^{1-\beta}/4$ to $h(x_0, y_1)$ can be computed by Chernoff bounds to be at most $2^{-\Omega(m^{1 - 2\beta})}$. Taking a union bound over all such $h(x_0, y_1)$, we obtain
\[ \pr[U \in T_y] \leq 2^{m^{8\beta}\beta \log m} 2^{-\Omega(m^{1 - 2\beta})} = 2^{-\Omega(m^{1 - 2\beta})} \]
as long as $\beta < 1/10$ and $m$ is sufficiently large.
\end{proof}

From these two claims, it follows that for a $1-e^{-m^{3\beta}/2}$ choices of $y$,
\begin{equation}
\label{eqn:statX}
\pr_X[h(X+y) \in T_y] - \pr_U[U \in T_y] \geq 1 - 1-e^{-m^{3\beta}/2} - 2^{-\Omega(m^{1 - 2\beta})}. 
\end{equation}
To finish the proof, we show how to replace $X$ with another variable $\Xbar$ of min-entropy at least $1.5m$ that is statistically close to it. We define $\Xbar$ as follows: First, choose $X$ from the $p$-biased distribution. If the Hamming weight of $X$ is at least $0.9pn$, set $\Xbar = X$. Otherwise, let $\Xbar$ be uniformly random in $\B^n$. We prove the following claim in Appendix~\ref{app:minent}:

\begin{claim}
\label{claim:minent}
$\Xbar$ has min-entropy at least $1.5m$. 
\end{claim}

Clearly the same conclusion holds for the distribution $\Xbar + y$. The statistical distance between $X$ and $\overline{X}$ is upper bounded by the probability that $X$ has Hamming weight less than $0.9pn$. By Chernoff bounds, the probability of this is at most $\exp(-\Omega(pn))$, which using the lower bound (\ref{eqn:pbounds}) is at least $\exp(-\Omega(m/\log(n/m))) = \exp(-m^{\Omega(1)})$ (since $m \geq n^{-0.99}$). Applying the triangle inequality, for all $y$ satisfying (\ref{eqn:statX}) we have 
\[ \pr_{\Xbar}[h(\Xbar+y) \in T_y] - \pr_U[U \in T_y] \geq 1 -  e^{-m^{\Omega(1)}}. \]
We conclude that the expected statistical distance between $h(\overline{X}+y)$ and $U_m$ for a random choice of $y$ is at least $(1-  e^{-m^{\Omega(1)}})(1- e^{-m^{\Omega(1)}})=1- e^{-m^{\Omega(1)}}$.

\section{Proof of Theorem~\ref{thm:upper2}}
\label{sec:upper2}

As in the proof of Theorem~\ref{thm:upper}, it is sufficient to show for every set $S$ of size $2^k$ and every $x_0$ in $S$ and $r_0$ in $\B^s$,
\[ \Pr_{M, X, R}[MX + BR = Mx_0 + Br_0] \leq \frac{1+(1/\delta)\cdot 2^{-k-s+m}}{2^m}\]
where the probability is taken over the random matrix $M$, $X$ chosen uniformly from $S$ and $R$ chosen uniformly from $\B^s$.

Assume that $p \leq 1/2$.  Let $S_0$ be the set $\{x + x_0: x \in S\}$. Then
\[
\Pr_{M, X, R}[MX + BR = Mx_0 + Br_0]   = \Pr_{M, X, R}[M(X + x_0)= B(R+r_0)] =\Pr_{M, Y, R}[MY = BR]
\]
where $Y$ is a random element from $S_0$. Let $M_i,B_i$ denote the $i$th row of $M$ and $B$.  Then 
\begin{multline*}
\Pr[MY = BR]  = \E_{M,Y,R} [\prod_{i=1}^{m}\frac{1+(-1)^{M_iY + B_iR}} {2}] \\
  = \frac{1}{2^m} \sum_{T\subseteq [m]} \E_{M, Y, R}\bigl[(-1)^{\sum_{i\in T}M_iY + B_iR}\bigr] 
 = \frac{1}{2^m} \sum_{T\subseteq [m]} \E_{M, Y}[(-1)^{\sum_{i\in T}M_iY}]\E_R[(-1)^{\sum_{i\in T}B_iR}].
\end{multline*}
Since any $t=m/2K$ rows of $B$ are linearly independent, for every nonempty $T$ of size at most $t$, $\sum_{i\in T} B_i \neq 0$ and so $\E[(-1)^{\sum_{i\in T}B_iR}]=0$.  On the other hand for every $T$ of size at least $t$ we have
\begin{multline*}
\E_{M, Y}[(-1)^{\sum_{i\in T}M_iY}] 
= \frac{1}{2^k}\sum_{y\in S_0}\E_M[(-1)^{\sum_{i\in T}M_iy}] 
= \frac{1}{2^k} \sum_{y\in S_0}(1-2p)^{|y|\cdot\/|T|} \\
\leq \frac{1}{2^k}\sum_{y\in \B^n}(1-2p)^{t|y|} 
\leq \frac{1}{2^k}(1+(1-2p)^t)^n
\leq \frac{e^{ne^{-2pt}}}{2^k}.
\end{multline*}
Since $B$ has full rank, the condition $\sum_{i \in T} B_i = 0$ is satisfied for at most $2^{m-s}$ sets $T$. Hence,
\begin{align*}
\sum_{T\subseteq [m]} \E_{M,Y}[(-1)^{\sum_{i\in T}M_iY}] \E_R[(-1)^{\sum_{i\in T}B_iR}] 
&= 1+\sum_{T:|T|>t} \E_{M,Y}[(-1)^{\sum_{i\in T}M_iy}] \E_R[(-1)^{\sum_{i\in T}B_is}] \\
& \leq 1+\sum_{T:|T|>t}  \frac{e^{ne^{-2pt}}}{2^k} \abs{\E_R[(-1)^{\sum_{i\in T}B_iR}]}  \\
& \leq 1+ 2^{m-s} \cdot \frac{e^{ne^{-2pt}}}{2^k}  \\
&= 1+ \frac{e^{ne^{-2pt}}}{2^{s+k-m}}.
\end{align*}
Plugging in $t= \frac{m}{2K}$ and $p =\frac{K}{m} \cdot \ln\frac{n}{\ln{c}}$ we get the desired bound.

\section{Local pseudorandom generators from local one-way functions}

\label{sec:applications}
\label{sec:prg}
A sequence of correlated random variables $X_1, \dots, X_m$ taking values in $\B^n$ has {\em $(s, \eps)$ conditional pseudo-min-entropy} $r$ if for every $1 \leq i \leq m$, there exists a random variable $Y_i$ jointly distributed with $X_1, \dots, X_{i-1}$ such that the min-entropy of $Y$ conditioned on any choice of $X_1, \dots, X_{i-1}$ is at least $r$ and for every circuit $D$ of size $s$,
\[ \bigabs{\pr[D(X_i) \mid X_1, \dots, X_{i-1}] - \pr[D(Y_i) \mid X_1, \dots, X_{i-1}]} \leq \eps. \]

Vadhan and Zheng~\cite{VZ} give the following construction of conditional pseudo-min-entropy sequences from a one-way function $f\colon \B^n \to \B^n$. Let $z_i, 1 \leq i \leq t$ be the random strings
\[ z_i = \/_{o_i}\lfloor f(x_{i1})  \circ \/x_{i1} \circ \dots \circ  f(x_{ik})  \circ x_{ik}\rfloor_{n-o_i} \]
where $1 \leq o_i \leq n$ is an offset, $x_{i1}, \dots, x_{ik}$ are random strings, and $_{f}\lfloor\/y\rfloor_\ell$ denotes truncating the first $f$ and last $\ell$ bits of $y$ respectively. Let $X_j = z_{1j}z_{2j}\dots\/z_{tj}$, where $z_{ij}$ denotes the $j$th bit of $z_i$. Vadhan and Zheng prove the following theorem (we state it in the nonuniform setting).

\begin{theorem}[Vadhan and Zheng] 
Suppose $f\colon \B^n \to \B^n$ is computable by a circuit of size $\poly(n)$ and is hard to invert on a $1/s$ fraction of inputs by circuits of size $s$. 
There exists offsets $o_1, \dots, o_t$ such that for every $\eps$, $X_1, \dots, X_m$ has $(s^{\Omega(1)}/\poly(n\eps), \eps)$ conditional pseudo-min-entropy at least $t(1/2 + \Omega((\log s)/n))$ where $k = O(n/\log{s})$, $t = O((n/\log{s})^2\log^2{n}\log{(1/\eps)})$ and $m=2(k-1)n$.
\end{theorem}

\iffalse
\begin{theorem}[Haitner, Reingold, and Vadhan] 
\label{theorem:entropytomin}
Suppose $G_{nb}\colon \B^n \to \B^m$ is computable by a circuit of size $\poly(n)$ such that $G_{nb}$ has $(T,\eps)$ next block pseudoentropy at least $n+\Delta$. For every $\eps$, set $u = O(n/\Delta)$ and $t = O((m/\Delta)^2\log^2{n}\log{(1/\eps')})$. Then there exist offsets $o_1, \dots, o_t$ such that , $X_1, \dots, X_{(u-1)m}$ has $(T-O(umt),  t^2(u\eps+\eps'+ 2^{-\Omega(t)}))$ next-block pseudo-min-entropy at least $t(n/m+\Omega(\Delta/m))$.
\end{theorem}
\fi

The following claim was proved in the uniform setting by Haitner, Reingold, and Vadhan. We need a nonuniform version of it, whose proof is analogous. We include it at the end of this section for completeness.

\begin{claim}
\label{claim:owfhash}
Suppose $X_1, \dots, X_m$ (where $X_i$ takes values in $\B^t$) has $(T,\eps_1)$ conditional pseudo-min-entropy $\alpha$. Let $H$ be an extractor family for min-entropy $\alpha$ with error $\eps_2$ so that every function in $H$ is computable in size $T_0$. Then with probability at least $1/2$ over the choice of $H$ the distribution $(H(X_1, R_1), \dots, H(X_m, R_m))$ is $(T - mT_0, m\eps_1+2m^2\eps_2)$-pseudorandom where $R_1,\dots, R_m \sim \B^r$.
\end{claim}

Instantiating Claim~\ref{claim:owfhash} with the function family from Theorem~\ref{thm:upper2} where we set the output length of  function to be  $t(1/2+\Omega((\log{s})/n))+r$ and let $d$ be the entropy loss. We obtain the following consequence for the function 
\[ G(x_{11}, \dots, x_{tk}, r_1,\dots, r_{m}) = (H(X_1, r_1 ), \dots, H(X_{m}, r_{m})). \]

\begin{proposition}
\label{prop:owfprg}
Suppose $f\colon \B^n \to \B^n$ is an $\ell$-local function computable by a circuit of size $\poly(n)$ and is hard to invert on a $1/s$ fraction of inputs by circuits of size $s$.  With probability at least $1/2$ over the choice of $H$, $G\colon \B^{nkt+ mr} \to \B^{nkt(1 + \Omega((\log s)/n))+mr}$ is an $O(\ell \cdot (nkt\ln{({t/\ln{c}})}+ mr))$-sparse, 
$(s^{\Omega(1)}/\poly(n\eps), \poly(n)(\eps + \sqrt{c}2^{-d/2}))$ pseudorandom generator where $k = O(n/\log{s})$, $t = O((n/\log{s})^2\log^2{n}\log{(1/\eps)})$ and $m=2(k-1)n$.
\end{proposition}
\iffalse
\begin{proposition}
\label{prop:owfprg}
Suppose $G_{nb}\colon \B^n \to \B^m$ is computable by a circuit of size $\poly(n)$ and $\ell$-local such that $G_{nb}$ has $(T,\eps)$ next block pseudoentropy at least $n+\Delta$. For every $\eps'$, set $u = O(n/\Delta)$ and $t = O((m/\Delta)^2\log^2{n}\log{(1/\eps')})$.   With probability at least $1/2$ over the choice of $H$, 
    \[G\colon \B^{nut+(u-1)ns} \to \B^{nut(1 + \Omega(\Delta/n))+(u-1)ns}\]
 is an $O(\ell \log(n/\log{1/\eps'}))$-local, $(T-n^{O(1)}, n^{O(1)}(\eps+\eps'))$ pseudorandom generator.
\end{proposition}
Vadhan and Zheng gave simple construction of next-bit pseudoentropy generator which is a next block pseudoentropy generator with block length $1$ bit from any one way function.
\begin{theorem} [Vadhan and  Zheng]
Let $f\colon\B^n\rightarrow\B^n$ be $(T,\gamma)$ one-way. Then $(f(U_n),U_n)$ has $(t',\eps)$ next-bit pseudoentropy at least $n+\log(1/\gamma)-\eps$, for $t'=t^{\Omega(1)}/\poly(n, 1/\eps)$.
\end{theorem}
Combining  with Proposition~ref{prop:owfprg}, we can get local pseudorandom generator from local one way function. In particular, we can get linear strech pseudorandom generator from exponiential hard one way function which can be computed by NC0.
\fi

We can improve the locality of $G$ at the expense of increasing its input and output length by the factor of $O(\ln(t/\ln{c}))$ via the following transformation of Applebaum, Ishai, and Kushilevitz. For every output of $G$, which is obtained by applying a sparse linear transformation to some $X_j$ and therefore has the form
\[ X_{jk_1} + \dots + X_{jk_t} \]
introduce auxiliary new inputs $r_{j3}, r_{j4}, \dots, r_{j(t-1)}$ for $G$ and replace its corresponding output by the tuple
\[ (X_{jk_1} + X_{jk_2} + r_{j3}, r_{j3} + X_{jk3} + r_{j4}, \dots, r_{j_{t-1}} + X_{j(t-1)} + X_{jt}). \]
Call this new function $G'$. Applebaum et al. show that if $G$ is $(s^{\Omega(1)}, s^{-\Omega(1)})$-pseudorandom, so is $G'$. Since every bit of $X_j$ comes either from some input $x_i$ or from some output $f(x_i)$, it follows that if $f$ has locality $\ell$, then $G'$ has locality $3\ell$.

\begin{proof}[Proof of Claim~\ref{claim:owfhash}]
Let $Y_i$ be the conditional min-entropy model for $X_i$. We consider the hybrid distributions
\begin{align*}
X^{(i)} &= (H(X_1, R_1), \dots, H(X_{i-1}, R_{i-1}), H(X_i, R_{i}), U_{i+1}, \dots, U_m) \quad\text{and} \\
Y^{(i)} &= (H(X_1, R_1), \dots, H(X_{i-1}, R_{i-1}), H(Y_i, R_{i}), U_{i+1}, \dots, U_m)
\end{align*}
where $U_1, \dots, U_m$ are uniformly random and independent. By the definition of conditional pseudo-min-entropy, for every $i$ the distributions $X^{(i)}$ and $Y^{(i)}$ are $(T - mT_0,  \eps_1)$-indistinguishable. Because $\calH$ is an extractor family, the distributions $(H, H(Y_i, R_i) \mid X_1, \dots, X_{i-1})$ and $(H, U_i)$ are within statistical distance $\eps_2$ for any choice of $X_1, \dots, X_{i-1}$. It follows that $(H, Y^{(i)})$ and $(H, X^{(i-1)})$ are within statistical distance at most $\eps_2$, so by Markov's inequality $Y^{(i)}$ and $X^{(i-1)}$ are within statistical distance $2m\eps_2$ with probability at least $1 - 1/2m$ over the choice of $H$. By a union bound, with probability at least $1/2$ over the choice of $H$, $Y^{(i)}$ and $X^{(i-1)}$ are $2m\eps_2$-statistically close for all $i$. For such a choice of $H$, by the triangle inequality $X^{(m)}$ is $(T - mT_0, m\eps_1+2m^2\eps_2)$ indistinguishable from $X^{(0)}$. Since $X^{(m)} = (H(X_1，R_1), \dots, H(X_m, R_m))$ and $X^{(0)}$ is the uniform distribution, we obtain the desired conclusion.
\end{proof}

\medskip
\paragraph{Acknowledgments} We would like to thank Rafail Ostrovsky for asking about the existence of a local extraction alternative to hash functions, and Elchanan Mossel, Oded Regev, Alon Rosen, Yuval Ishai, and Sid Jaggi for helpful discussions.

\bibliographystyle{alpha}
\bibliography{database}

\appendix

\section{Statistical distance versus collision probability}
\label{app:colldist}

Here we give a standard bound of the statistical distance of the distributions $(H, H(X))$ and $(H, U)$ in terms of the collision probability of $H$. The proof is very well known when the distribution over $H$ is uniform. In our application, however, it is not, so we include the proof for completeness.

\begin{claim}
\label{claim:colldist}
Let $H$ be a distribution on functions from $\B^n$ to $\B^m$, $X$ be a distribution over $\B^n$ and $U$ be the uniform distribution over $\B^m$. The statistical distance between $(H, H(X))$ and $(H, U)$ is at most
\[ \frac12 \sqrt{2^k\pr_{H, X, X'}[H(X) = H(X')] - 1}. \]
where $X$ and $X'$ are independent samples from the same distribution.
\end{claim}
\begin{proof}
We upper bound the $\ell_1$ distance, which is twice the statistical distance:
\begin{align}
\ell_1((H, h(X)), (H, U))
&= \sum\nolimits_{h, b} \bigabs{\pr_{H, X}[H = h \wedge H(X) = b] - \pr_H[H = h]2^{-m}]} \notag \\
&= \sum\nolimits_{h, b} \sqrt{\pr_H[H = h]} \cdot \sqrt{\pr_H[H = h]}\Bigl|\pr_X[h(X) = b] - 2^{-m}]\Bigr| \notag \\
&\leq \sqrt{\sum\nolimits_{h, b} \pr_H[H = h]} \cdot \sqrt{\sum\nolimits_{h, b} \pr_h[H = h]\bigl(\pr_{H, X}[H(X) = b] - 2^{-m}]\bigr)^2} \notag \\
&= \sqrt{2^m} \cdot \sqrt{\pr_{H, X, X'}[H(X) = H(X')] - 2^{-m}} \notag \\
&=  \sqrt{2^m\pr_{H, X, X'}[H(X) = H(X')] - 1}. \tag*{\qedhere}
\end{align}
\end{proof}

\section{Bounds on the inverse of entropy}
\label{app:entropy}

\begin{lemma}
\label{lemma:entropy}
For every $p \in (0, 1/2]$,
\[ \frac{H(p)}{6 \log_2 2/H(p)} \leq p \leq \frac{H(p)}{\log_2 1/H(p)}. \]
\end{lemma}

The upper bound on $p$ follows from the inequality $H(p) \geq p\log_2 1/p$. Applying twice we obtain
\[ \frac1p \geq \frac{1}{H(p)} \log_2 \frac1p \geq \frac{1}{H(p)} \log_2 \Bigl(\frac{1}{H(p)} \log_2 \frac1p\Bigr) 
\geq \frac{1}{H(p)} \log_2 \frac{1}{H(p)} \]
because $1/p \geq 2$. For the lower bound, we apply $H(p) \leq 2p\log_2 1/p$ twice to obtain
\[ \frac{1}{p} \leq \frac{2}{H(p)} \log_2 \frac{1}{p} \leq \frac{2}{H(p)} \log_2 \Bigl(\frac{2}{H(p)} \log_2 \frac{1}{p}\Bigr). \]
Now $2/H(p) \geq (1/p) \log_2(1/p) \geq \sqrt{\log_2(1/p)}$, which is true for every $p \in (0, 1]$. Therefore
\[ \frac{1}{p} \leq \frac{2}{H(p)} \log_2 \Bigl(\frac{8}{H(p)^3}\Bigr) = \frac{6}{H(p)} \log_2 \Bigl(\frac{2}{H(p)}\Bigr). 
\hfill \]

\section{Proof of Claim~\ref{claim:minent}}
\label{app:minent}

The maximum probability in $\Xbar$ is attained by those strings in $\B^n$ that have Hamming weight exactly $0.9pn$. Let $a$ be such a string. Then
\begin{align*}
\pr[\Xbar = a] &\leq \pr[X = a] + \pr[U = a] \\ &= p^{0.9pn}(1-p)^{1.1(1-p)n} + 2^{-n} \\
&\leq 2^{-nH(p)} \cdot p^{-0.1pn} + 2^{-n} \\
&= 2^{-nH(p)} \cdot 2^{0.1p\log_2(1/p)n} + 2^{-n} \\
&\leq 2^{-nH(p)} \cdot 2^{0.1nH(p)} + 2^{-n} \\
&= 2^{-0.9nH(p)} + 2^{-n} \\
&= 2^{-1.8m} + 2^{-n} \leq 2^{-1.5m}.
\end{align*}

\iffalse
\section{Proof of Proposition~\ref{prop:lower}}
\label{app:lower}

\section{Existence of matrix $B$}
\label{app:exis}
Notice that,  we want a low density parity check code (LPDC) ~\cite{Ga62} which has constant rate , linear minimum distance and full rank parity check matrix.  We call  LPDC of constant rate and linear minimum distance a good LPDC. First of all, we can modify any parity check matrix of good LPDC to have full rank by deleting some rows. 
\begin{lemma}
Suppose $s\times m$ ($s<m$) matrix $B$ is a  parity check matrix.  There exists another parity check matrix $B'$ of full rank so that $B'$ and $B$  have the same minimum distance and rate.
\end{lemma}
\begin{proof}
Let $B'$ be $B$, then deleting rows of $B'$ to make $B'$ full rank and rank$(B)$ = rank $(B')$. For any $x\in\B^m$, we can know $Bx=0\Rightarrow B'x=0$ and  $B'x=0\Rightarrow Bx=0$ by the construction of $B'$. So that they have the same minimum distance and rate.
\end{proof}
 So we can get $B$ from any good LDPC.  One simple construction of good LDPC is from expander graphs:   let adjacent  matrix of expander graph be the parity check matrix.   Following theorem  from ~\cite{SS96}  implies the resulting code is good LDPC.
\begin{theorem}
Let $G$ be a $(c,d,\alpha,c/2)$-graph, i.e. $(c,d)$-graph ($d>c$) with $n$ vertices on the left each of degree $c$, $(c/d)n$ vertices on the right each of degree $d$ and for each subset $S$ with $|S|\leq \alpha n$ then $|N(S)|\geq c/2|S|$. Then the resulting code has rate $1-c/d$ and minimum relative distance $\alpha$.
\end{theorem}
\fi

\end{document}